\documentclass[useAMS,usenatbib,margin=1in]{article}
 \usepackage{amsmath,natbib,amssymb}
\usepackage{bm} \usepackage{bbm}

\usepackage[normalem]{ulem}

\usepackage{multirow}
\usepackage{geometry}
\renewcommand{\Pr}[1]{\mathbb{P}\!\left(#1\right)}
\newcommand{\Prob}[1]{\mbox{Pr}\!\left(#1\right)}
\newcommand{\Var}[1]{\mbox{Var}\!\left(#1\right)}
\newcommand{\x}{{\mathbf{x}}} 
 \newcommand{\A}{{\mathbf{A}}}

\newcommand{\p}{{\mathbf{p}}} 
\newcommand{\R}{{\mathbf{R}}} \newcommand{\br}{{\mathbf{r}}}
 
\newcommand{\U}{{\mathbf{U}}} 
 
\newcommand{\X}{{\mathbf{X}}} 
\newcommand{\Z}{{\mathbf{Z}}}

\newcommand{\1}{{\mathbf{1}}}

\usepackage{verbatim} 
\usepackage{enumitem} 
\usepackage{appendix} \usepackage{pdfpages}
\usepackage[usenames,dvipsnames]{pstricks} \usepackage{epsfig}
\usepackage{pst-grad} 

\newtheorem{theorem}{Theorem}

\newtheorem{definition}[theorem]{Definition}

\newtheorem{remark}[theorem]{Remark}

\newenvironment{proof}[1][Proof]{\noindent\textbf{#1.} }{\ \rule{0.5em}{0.5em}}

\usepackage{fancyhdr}
\begin{document}

\title{Sensitivity Analysis for matched pair analysis of binary data: From worst case to average case analysis}

\author{Raiden B. Hasegawa\textsf{\footnote{\textit{Address for correspondence}: Department of
Statistics, The Wharton School, University of Pennsylvania, Philadelphia, PA
19104-6340 US, Email: raiden@wharton.upenn.edu.}} \ and \  Dylan S. Small}

\date{May 16, 2018}

\maketitle

\begin{abstract} In matched observational studies where treatment assignment is
not randomized, sensitivity analysis helps investigators determine how sensitive
their estimated treatment effect is to some unmeasured confounder. The standard
approach calibrates the sensitivity analysis according to the worst case bias in
a pair.  This approach will result in a conservative sensitivity analysis if the
worst case bias does not hold in every pair. In this paper, we show that for binary
data, the standard approach can be calibrated in terms of the average bias in a
pair rather than worst case bias.  When the worst case bias and average bias differ, the average bias
interpretation results in a less conservative sensitivity analysis and more
power. In many studies, the average case calibration may also carry a more
natural interpretation than the worst case calibration and may also allow researchers to incorporate additional data to establish an empirical basis with which to calibrate a sensitivity analysis. We illustrate this with a study of the effects of cellphone use on the incidence of automobile accidents. Finally, we extend the average case calibration to the sensitivity analysis of confidence intervals for attributable effects.  \newline
\textbf{keywords:} attributable effects; binary data; causal inference; cellphone; 
majorization; sensitivity analysis; traffic collision.
\end{abstract}

\label{firstpage}
\thispagestyle{fancy}

\section{INTRODUCTION}\label{sec:intro}

\subsection{Sensitivity analysis as causal evidence}\label{subsec:evidence}

 In matched-pair observational studies, causal conclusions based on usual inferential methods (e.g., McNemar's test for binary data) rest on the assumption that matching on observed covariates has the same effect as randomization (i.e., that there are no unmeasured confounders). In other words, it is assumed that there are no unobserved covariates relevant to both treatment assignment and outcome. A sensitivity analysis assesses the sensitivity of results to violations of this assumption. \cite{cornfield1959} introduced a model for sensitivity analysis that was a major conceptual advance in the field of observational studies. A modern approach to sensitivity analysis is introduced in \cite{rosenbaum1987}; Rosenbaum's approach builds on Cornfield's model (\cite{cornfield1959}) but incorporates uncertainty due to sampling variance. There are other contemporary sensitivity analysis models, see for example \cite{mccandless2007} for a Bayesian approach, but we restrict our focus to Rosenbaum's approach. Rosenbaum's sensitivity analysis yields an upper limit on the magnitude of bias to which the result of the researcher's test of no treatment effect is insensitive for a given significance level $\alpha$. More specifically, \cite{rosenbaum1987} derives bounds on the p-value of this test given an upper bound, $\Gamma$, on the odds ratio of treatment assignment for a pair of subjects matched on observed covariates. $\Gamma$ can be thought of as a measure of ``worst case'' bias in the sense that treatment assignment probabilities in matched pairs are allowed to vary arbitrarily as long as the odds ratio of treatment assignment for a pair of subjects is no greater than $\Gamma$. The largest $\Gamma$ for which the p-value is less than 0.05 is denoted by $\Gamma_{sens}$. We will use $\Gamma_{truth}$ to distinguish the true unknown worst case bias. $\Gamma_{sens}$ is interpreted in Rosenbaum's sensitivity analysis as the largest value of the worst case bias across matched pairs that does not invalidate the finding of evidence for a treatment effect. We refer to this as a \textit{worst case calibrated} sensitivity analysis. 
A classic example of this type of analysis is given in Chapter 4 of \cite{rosenbaum2002}. Applying the worst case sensitivity analysis to a study of the effects of heavy smoking on lung cancer mortality (\cite{hammond1964}), Rosenbaum finds that $\Gamma_{sens} \approx 6$ and interprets this result cogently:

\begin{quote}
  To attribute the higher rate of death from lung cancer to an unobserved covariate rather than to an effect of smoking, that unobserved covariate would need to produce a sixfold increase in the odds of smoking, and it would need to be a near perfect predictor of lung cancer. 
\end{quote}
A brief, more formal review of Rosenbaum's sensitivity analysis framework is in Section \ref{ss-review}. 

The worst case calibrated sensitivity analysis raises several potential questions. If we are convinced that there is no pair in Hammond's smoking study such that one unit is more than six times as likely to smoke as the other (i.e., $\Gamma_{truth}\le \Gamma_{sens}$), then we would conclude that our study provides convincing evidence that heavy smoking increases the rate of lung cancer mortality. However, what if, on average, unmeasured confounders do not alter the odds of smoking greatly but there are some subjects for whom the unmeasured confounders make them almost certain to smoke, e.g., a subject who experiences huge peer pressure to smoke.
If such a subject ends up in our sample of matched pairs, and we condition on matched pairs in which only one unit receives treatment, a standard practice when conducting matched pair randomization tests, then the odds ratio of treatment assignment in the matched pair containing that subject, and consequently $\Gamma_{truth}$, will be infinite. In such a case, since $\Gamma_{sens}$ is generally finite, we'd expect it to be smaller than $\Gamma_{truth}$. Now, suppose that there are such pairs in the Hammond study but that for most pairs the odds ratio of smoking between the units is much smaller than six. Using the worst case calibrated sensitivity analysis, we would conclude that the study is sensitive to bias. Is there potentially some natural quantification of average bias over the sample of matched pairs, say, $\Gamma'_{truth}$, that isn't infinite and perhaps is smaller than six? And if we calibrate our sensitivity analysis to this measure of bias rather than the worst case measure, will the sensitivity analysis be valid in the sense that the inference is conservative at level $\alpha$ for any $\Gamma \ge \Gamma'_{truth}$? If it is valid, are there other advantages to using the \textit{average case calibrated} sensitivity analysis over the worst case calibrated sensitivity analysis? In what follows, we attempt to answer these motivating questions in the context of a matched pair analysis of the association between cellphone use and car accidents.

\subsection{Outline}

In this paper we demonstrate that interpreting sensitivity analysis results in terms of average case rather than worst case hidden bias is both valid and conceptually more natural in many common scenarios. To illustrate our claim that the average case analysis is more natural we will perform a causal analysis of a study by \cite{tibshirani1997} that asks if there is an association between cellphone use and motor-vehicle collisions. The study is described in the following section. In section \ref{review} we review the model for sensitivity analysis of tests of no treatment effect and sensitivity intervals for attributable effects for binary data. In section \ref{worst-avg} we discuss the theory behind the validity of average case sensitivity analysis. Finally, the \cite{tibshirani1997} study is examined in this new light in section \ref{examples}. In particular, we see how the average case sensitivity analysis makes it possible to use additional information from the problem to empirically calibrate our sensitivity analysis in Section \ref{subsec:inter} and we extend the average case sensitivity analysis to the study of sensitivity intervals for attributable effects in Section \ref{attr-sec}.

\subsection{Motivating Example: Effects of cellphone use on the incidence of motor-vehicle collisions}

\cite{tibshirani1997} conducted a case-crossover study of the effects of cellphone use on the incidence of car collisions. In a case-crossover study each subject acts as her own control which has the benefit of controlling for potential confounders that are time-invariant, even if they are unobserved. Data collection took place at a collision reporting center in Toronto between July 1, 1994 and August 31, 1995 during weekday peak hours (10 AM to 6 PM). Consenting drivers who reported having been in a collision with substantial property damage and who owned a cellphone were included in the study. Drivers involved in collisions that involved injury, criminal activity, or transport of dangerous goods were excluded. The resulting study population included 699 individuals who gave permission to review their cellphone records and filled out a brief questionnaire about their personal characteristics and the features of the collision. The matched pair analysis compared cellphone usage in the 10-minute hazard window prior to the crash with a 10-minute control window on a chosen day prior to the crash. We will denote the time of the crash as $t$ and the hazard window as $t-10$ to $t-1$ minutes. The authors examined several different control windows:
\begin{enumerate}[leftmargin=20pt]
  \item{} \textit{Previous day}: time $t-10$ to $t-1$ minutes on the previous day.
  \item{} \textit{Previous weekday/weekend}: time $t-10$ to $t-1$ minutes on the previous weekday if the crash took place on a weekday and similarly if the crash took place on a weekend.
  \item{} \textit{One week prior}: time $t-10$ to $t-1$ minutes one week prior to the collision.
  \item{} \textit{Busiest cellphone day of previous three days}: time $t-10$ to $t-1$ minutes on the one day among the prior three to the collision with the most cellphone calls.
\end{enumerate} 

For each choice of control window, \cite{tibshirani1997} found that there was a significant positive association between cellphone usage and traffic collision incidence. The 2 x 2 contingency tables shown in Table \ref{tbl-cw} summarize the data using the four different control windows. 

\begin{table}[h!]
\centering
\begin{tabular}{lrrr}

&&\multicolumn{2}{c}{Control} \\
&& On phone & Not on phone \\ 
  \hline
  &&\multicolumn{2}{c}{\textit{Previous Weekday/end}} \\
  \hline
\multirow{2}{*}{Hazard} & On phone & 12 & 158 \\ 
& Not on phone & 23 & 506 \\ 
   \hline
  &&\multicolumn{2}{c}{\textit{One Week Prior}} \\
  \hline
\multirow{2}{*}{Hazard} & On phone & 6 & 164 \\ 
& Not on phone & 21 & 508 \\ 
   \hline
&&\multicolumn{2}{c}{\textit{Previous Driving Day}} \\
  \hline
\multirow{2}{*}{Hazard} & On phone & 18 & 119 \\ 
& Not on phone & 20 & 171 \\ 
   \hline
&&\multicolumn{2}{c}{\textit{Most Active Cellphone Day}} \\
  \hline
\multirow{2}{*}{Hazard} & On phone & 17 & 135 \\ 
& Not on phone & 43 & 504 
\end{tabular}
\vspace{20pt}
\caption{\textbf{One Week Prior}: results for one week prior control window versus hazard window;\textbf{Previous Weekday/end}: results for previous weekday/weekend control window versus hazard window; \textbf{Previous Driving Day}: results for previous driving day control window versus hazard window;\textbf{Most Active Cellphone Day}: results for most active cellphone day in previous 3 days control window versus hazard window.}\label{tbl-cw}
\end{table}

\subsection{Sensitivity of results to hidden bias}\label{sens-examp}
As this was an observational study, the associations cannot be assumed to be causal. We would like to quantify how large a hidden bias would have to be to explain the observed association between cellphone use and car accidents without it being causal. A sensitivity analysis seems appropriate and is a straightforward exercise (see Chapter 4, \cite{rosenbaum2002} for example). Table \ref{sens-table} shows the results of a standard worst case sensitivity analysis for each control window. Here, $\Gamma_{sens}$ is the largest value of $\Gamma$ such that the result are still significant at the $\alpha=0.05$ level. In our analysis of the case-crossover study from \cite{tibshirani1997} we condition on subjects who were on a cellphone in exactly one of the control and hazard windows (i.e., discordant case-crossover pairs). Thus, the odds ratio of treatment assignment for the two windows observed for any case-crossover subject can be viewed as the conditional odds that treatment occurs in a particular window. Hence, we can interpret $\Gamma$ as the maximum (and $1/\Gamma$ as the minimum) over all study subjects of the odds that a driver is using a cellphone during the hazard window and not during the control window.

\begin{table}[ht]
\centering
\begin{tabular}{lr}
 Control Window & $\Gamma_{sens}$ \\ 
  \hline
previous weekday/weekend & 4.92 \\ 
one week prior & 5.53 \\
  previous driving day & 4.15 \\  
  most active cellphone day & 2.40
\end{tabular}
\vspace{20pt}
\caption{Sensitivity analysis for (marginal) $\alpha=0.05$.}\label{sens-table}
\end{table}

The sensitivity analysis suggests that the most active cellphone day control window was the most conservative analysis. This is unsurprising since we would expect that the treatment assignment (cellphone use) would be biased toward the control window on a day when you used a cellphone relatively often. We can interpret these results as follows: \textit{the observed ostensible effect is insensitive to hidden bias that increases the odds that a driver was on a cellphone in the hazard window and not the control window on the most active cellphone day by at most a factor of 2.4.} In many observational studies this type of statement is very useful. However, it may be plausible that some study participants are exposed to infinite (or at least very large) hidden bias. For example, this happens if a subject was not driving during the control window and (almost) always uses her landline rather than her cellphone when she is not driving. When we condition on case-crossover pairs where the treatment is received in exactly one of the windows -- a standard practice when conducting a matched pair randomization test -- such a driver is always on a cellphone during the hazard window. When this happens, the observed ostensible effect is (almost) always sensitive to hidden bias, no matter how strong the observed association. Implicitly, in the worst case sensitivity analysis, the investigator is supremely skeptical; she assumes that it could be that all study participants suffer from the worst case hidden bias which, when it is possible that some study participant suffers from unbounded hidden bias, renders sensitivity analysis under the standard worst case interpretation uninformative. Yet in many studies where unbounded hidden bias in some matched pairs is plausible, as in our motivating example, we still want to examine the sensitivity of our results to potential hidden bias. If we could perform a valid, average case calibrated sensitivity analysis then we could (1) make sensitivity analysis informative even in the presence of pairs subject to unbounded hidden bias and (2) make the interpretation of sensitivity analysis results far less conservative. It turns out that there is a measure of the sample average bias that is generally finite in the presence of pairs subject to unbounded bias for data with binary treatment and outcome. Moreover, the sensitivity analysis calibrated to this measure of average bias is valid when using McNemar's statistic to test the null hypothesis of no treatment effect against the alternative of a positive treatment effect (i.e., that talking on a cellphone while driving increases the rate of automobile accidents).
 
\section{NOTATION AND REVIEW}\label{review}
\subsection{Notation} 

Our study sample consists of $S$ matched pairs where each pair $s=1,2,\dots,S$ is matched on a set of observed relevant covariates $\x_{s1}=\x_{s2}=\x_s$. Units in each pair are indexed by $i=1,2$. We let $Z_{si}$ and $R_{si}$ denote the treatment assignment and outcome, respectively, of the $i$-th unit of the $s$-th pair. The potential outcomes under treatment and control are denoted as $r_{Tsi}$ and $r_{Csi}$, respectively. Hence, we can write $R_{si}=Z_{si}r_{Tsi} + (1-Z_{si})r_{Csi}$. Under Fisher's sharp null hypothesis of no treatment effect, i.e., $r_{Tsi}=r_{Csi}$ for all $i$, we have that $R_{si}=r_{Csi}$. Hereafter, we will work under the null hypothesis and under the assumption that each pair was matched on some set of observed covariates $\x_s$. Additionally, we assume that there is some unobserved covariate $U_{si}$ that is associated with both treatment assignment and outcome and let $u_{si}$ be the realization of $U_{si}$ for the $i$-th unit of the $s$-th pair. Within pair differences in treatment and outcome will be denoted as $V_s = Z_{s1}-Z_{s2}$ and $y_s = r_{Cs1}-r_{Cs2}$. It will be convenient to define the following vector quantities: $\Z = (Z_{11},Z_{12},\dots,Z_{S2})^T$, $\mathbf{r} = (r_{C11},r_{C12},\dots,r_{CS2})^T$,$\U = (U_{11},U_{12},\dots,U_{S2})^T$, and $\A = (|y_1|,|y_2|,\dots,|y_S|)^T$.

To be very clear about the information on which we are conditioning we will define some important information sets. Let $\mathcal{F} = \{(\x_s,u_{si},r_{Csi},r_{Tsi}):\;s=1,2\dots,S,\, i=1,2\}$ be the set of \textit{fixed} observed and unobserved covariates for all units. Let $\mathcal{Z} = \{\Z:\; |V_s|=1,\, s=1,\dots,S\}$ be the set of matched pairs such that only one unit receives treatment. We assume that $\R$ is binary and we define $\mathcal{A}_{\1}=\{\A:\; |y_s|=1,\, s=1,\dots,S\}$. So $\mathcal{Z}\cap\mathcal{A}_{\1}$ is the set of discordant matched pairs. In the analysis that follows, we will condition on $\mathcal{F},\,\mathcal{Z}\cap\mathcal{A}_{\1}$. 
\subsection{Review: sensitivity analysis for binary data}\label{ss-review}

Under the assumption that all variables that confound treatment assignment are observed,
\[ Z_{si} {\perp\!\!\!\perp} (r_{Csi},r_{Tsi})\,|\,\X_s \quad \text{(Ignorability)}\]
our matched observational study should closely resemble a randomized study and thus  $\Pr{\Z=\mathbf{z}|\,\mathcal{F},\,\mathcal{Z}\cap\mathcal{A}_{\1}} = 1/2^S$ for $\mathbf{z}\in\mathcal{Z}$. In practice, this assumption is rarely valid and the probability of treatment assignment depends materially on the unobserved covariates $\U$. A second assumption made in the causal framework introduced in \cite{rosenbaum1983} is the \textit{Positivity} assumption -- $0 < \Pr{Z_{si}=1|\,\X_s} <1$ for all $s=1,2,\dots,S$ and $i=1,2$ -- which says that all units have a chance of receiving treatment. In our case-crossover study, however, this may not be an appropriate assumption. We introduce an example of how our case-crossover study might violate the positivity assumption in Section \ref{subsec:inter} and how our average case sensitivity analysis framework is able to handle violations of positivity. 

When both $Z$ and $r$ are binary it is common to use McNemar's statistic to test for treatment effect:
\begin{definition}
  \label{defn:mcnemars}
  For a matched pair study with binary treatment and outcome we define \textbf{McNemar's statistic} to be
  \begin{equation} \label{eqn:mcnemars} T(\Z,\br) = \sum_{s=1}^S \mathbbm{1}\{V_sY_s=1\}\,.\end{equation}
\end{definition}
Under the null distribution of no treatment effect $T(\Z,\br)$ follows a Poisson-Binomial distribution with probabilities $\{p_1,p_2,\dots,p_S\}$ where $p_s = \Pr{(Z_{s1}-Z_{s2})(r_{s1}-r_{s2})=1}$ is the probability that the unit with positive outcome, i.e., $r=1$, receives treatment in pair $s$. If we consider only discordant pairs and we assume, without loss of generality, that the first unit in each pair is the unit with positive outcome we may write
\begin{equation}
  \label{eq:ps1}
  p_s = \Pr{Z_{s1}=1|\mathcal{F},\,\mathcal{Z}\cap\mathcal{A}_{\mathbf{1}}}\,.
\end{equation}
Recall that the Poisson-Binomial distribution is the sum of independent, not necessarily identical Bernoulli trials. If $\X_s$ contains the complete set of relevant covariates then $p_s$ equals $1/2$ for all pairs and we can conduct inference using $\operatorname{B}(1/2,S)$ as our null distribution, effectively treating our data as being the outcome of a randomized study. As we mentioned earlier in this section, if there is some unobserved characteristic $U$ that is relevant to treatment assignment and outcome then $\{p_1,\dots,p_S\}$ are unknown and consequently the exact null distribution is no longer available to the investigator. When this is the case, a sensitivity analysis like the one conducted informally in Section \ref{sens-examp} can be used to determine how sensitive the investigator's conclusions are to departures from the ideal randomized design. Following Chapter 4 of \cite{rosenbaum2002} we can formalize the notion of a sensitivity analysis introduced in Sections \ref{subsec:evidence} and \ref{sens-examp} with a simple sensitivity model where
\begin{equation}
  \label{eq:sens-gamma}
  \frac{1}{1+\Gamma} \le \Pr{Z_{s1}=1|\mathcal{F},\,\mathcal{Z}\cap\mathcal{A}_{\mathbf{1}}} \le \frac{\Gamma}{1+\Gamma}
\end{equation}
for all $s=1,\dots,S$ and where $\Gamma\ge 1$ is the sensitivity parameter that bounds the extent of departure from a randomized study. Proposition 12 in Chapter 4 of \cite{rosenbaum2002} states that \eqref{eq:sens-gamma} is equivalent to the existence of the following model
\begin{equation}
  \label{eq:sens-equiv}
  \log\left(\frac{p_s}{1-p_s}\right) = \gamma\left(u_{s1}-u_{s2}\right)\,,\;s=1,\dots,S
\end{equation}
where $\exp(\gamma) = \Gamma$, $\gamma \ge 0$, and $u_{si} \in [0,1]$ for $s=1,\dots,S$ and $i=1,2$. The restriction of the unobserved confounder to the unit interval in this equivalent representation preserves the non-technical interpretation of $\Gamma$ used in section \ref{sens-examp} as a bound on the odds that the driver was talking on a cellphone in the hazard window. Henceforth, we assume that $U_{si}$ and its realization $u_{si}$ belongs to the unit interval for $s=1,\dots,S$ and $i=1,2$. However, the distribution of $U_{si}$ on the unit interval may be arbitrary.

Under this sensitivity model, if we let $T^+$ be binomial with success probability $\Gamma/(1+\Gamma)$ and $T^-$ be binomial with success probability $1/(1+\Gamma)$ it follows from Theorem 2 of \cite{rosenbaum1987} that
\begin{equation}
  \label{eq:stoch-order}
  \Pr{T^- \ge k} \le \Pr{T \ge k|\mathcal{F},\,\mathcal{Z}\cap\mathcal{A}_{\mathbf{1}}} \le \Pr{T^+ \ge k}
\end{equation}
for all $k = 1,\dots,S$. This inequality is tight in the sense that it holds for any realization $\mathbf{u}$ of $\mathbf{U}$. For conducting a hypothesis test, the stochastic ordering in \eqref{eq:stoch-order} gives us bounds on the p-value of our test for a given magnitude of bias $\Gamma$. If $\Gamma\ge \Gamma_{truth}$, then $T^+$ yields a valid, albeit conservative,  reference distribution for testing the null hypothesis of no treatment effect against the alternative of a positive treatment effect.

\subsection{Attributable effects for binary outcomes: hypothesis tests and confidence intervals}
\label{subsec:attr-binary}
Attributable effects are a way to measure the magnitude of a treatment effect on a binary outcome.  The number of attributable effects is the number of positive outcomes among treated subjects that would not have occurred if the subject was not exposed to treatment.  In this section, we review \cite{rosenbaum2002a}'s procedure to construct one-sided confidence statements about attributable effects in the context of the cellphone case-crossover study.

Let $\widetilde{S}$ be the number of \textit{all} pairs in the study, discordant or not, and let the first $S$ be the discordant pairs. If we assume that $r_{Tsi} \ge r_{Csi}$, that talking on a cellphone cannot prevent an accident, then we can write the attributed effect as
\begin{equation}
  A = \sum_{s=1}^{\widetilde{S}}\sum_{i=1}^2 Z_{si}(r_{Tsi}-r_{Csi}) = \sum_{s=1}^{\widetilde S} Z_{s1}(r_{Ts1}-r_{Cs1})
\end{equation}
where the first unit of $s$-th pair is the observation from the hazard window. Why does the second equality hold? If the subject was talking on a cellphone in the control window, that is $Z_{s2}=1$, then we observe $r_{Ts2}=0$ which by our assumption that talking on a cellphone cannot prevent an accident implies that $r_{Cs2}=0$. So attributable effects can only occur among discordant pairs where the subject was talking on a cellphone in the hazard window or concordant pairs where the subject was talking on a cellphone in both windows. The following table characterizes the four types of possible pairs in our case-crossover study,

\begin{table}[h!]
  \centering
  \begin{tabular}{lrrrrrr}
    & $Z_{s1}$ & $Z_{s2}$ & $R_{s1}$ & $R_{s2}$ & $r_{Ts1}$ & $r_{Cs1}$ \\
    \hline
    $D(+,-)$ & 1 & 0 & 1 & 0 & 1 & - \\
    $D(-,+)$ & 0 & 1 & 1 & 0 & 1 & 1 \\
    $C(-,-)$ & 0 & 0 & 1 & 0 & 1 & 1 \\
    $C(+,+)$ & 1 & 1 & 1 & 0 & 1 & - \\
  \end{tabular}
  \vspace{20pt}
  \caption{The four possible types of pairs in our case-crossover study. $D$ and $C$ indicate discordant and concordant pairs, respectively, and the $+$ and $-$ indicate if a unit in the pair was treated or not, respectively.}\label{tbl:attr}
\end{table}
$D$ and $C$ indicate discordant and concordant pairs, respectively. $D(+,-)$ is the set of discordant pairs where the subject was on a cellphone in the hazard window, $D(-,+)$ is the set of discordant pairs where the subject was on a cellphone in the control window, $C(+,+)$ is the set of concordant pairs where the subject was on a cellphone in both hazard and control windows, and $C(-,-)$ is the set of concordant pairs where the subject was not on a cellphone in either window. If there are no attributable effects then we know that $r_{Cs1}=1$ in $D(+,-)$ and $C(+,+)$ and we have that $R_{s1}=r_{Cs1}$ for all pairs $s$, concordant or discordant. We can write the probability that the subject was talking on a cellphone at the time of accident for each type of pair as (1) $\Pr{Z_{s1}R_{s1}=1 | D(+,-)\cup D(-,+)} = p_s$, where $p_s$ here is equivalent to the $p_s$ defined in Section \ref{ss-review} when there are no attributable effects; (2) $\Pr{Z_{s1}R_{s1}=1 | C(-,-)} = 0$; and (3) $\Pr{Z_{s1}R_{s1}=1 | C(+,+)} = 1$. Now let $c^+=|C(+,+)|$ denote the cardinality of the set of concordant pairs where the subject was on a cellphone in both windows and let $s=S+1,\dots,S+c^+$ be the pairs belonging to $C(+,+)$. Then if $A=0$ we can define the standardized deviate for McNemar's statistic $T$ as 
\begin{align}\label{eq:norm-dev}
  \widetilde{T} &=  \frac{\sum_{s=1}^{S} Z_{s1}r_{Cs1} - \sum_{s=1}^S p_s}{\left\{\sum_{s=1}^Sp_s(1-p_s)\right\}^{1/2}} \notag  \\
& =  \frac{\sum_{s=1}^{S+c^+} Z_{s1}R_{s1} - \left(\sum_{s=1}^S p_s + c^+\right)}{\left\{\sum_{s=1}^Sp_s(1-p_s)\right\}^{1/2}}\,. 
\end{align}
$\widetilde{T}$ defines a normal reference distribution for $\sum_{s=1}^{S}Z_{s1}r_{Cs1}$ that we can use to conduct approximate inference.
If $A=a>0$, then $Z_{\widetilde{s}1}R_{\widetilde{s}1} = Z_{\widetilde{s}1}r_{T\widetilde{s}1} = Z_{\widetilde{s}1}(r_{C\widetilde{s}1}+1)$ for pair $\widetilde{s}$ belonging to the set of $a$ pairs with attributable accidents and the second equality above does not hold. When this equality fails to hold, the standard normal deviate $\widetilde T$ cannot be computed from the observed data conditional on $\mathcal F$.
How then can we adjust $\widetilde{T}$ for attributable accidents so that it can be computed from the observed data? Because we've assumed talking on a cellphone cannot prevent an accident, we only need to consider two cases. If pair $\widetilde{s}$ belongs to $D(+,-)$ then we subtract $Z_{\widetilde{s}1}(r_{T\widetilde{s}1}-r_{C\widetilde{s}1})=1$ from $\sum_{s=1}^{\widetilde{S}} Z_{s1}R_{s1}$, $p_{\widetilde{s}}$ from the expectation, and $p_{\widetilde{s}}(1-p_{\widetilde{s}})$ from the variance term. If $\widetilde{s}$ belongs to $C(+,+)$ we again subtract $1$ from  $\sum_{s=1}^{\widetilde{S}} Z_{s1}R_{s1}$ and subtract $1$ from the $|C(+,+)|$ in the expectation while leaving the variance term unchanged.

Let $\boldsymbol{\delta} = (\delta_{11},\delta_{12},\dots,\delta_{\widetilde{S}1},\delta_{\widetilde{S}2})^T$ be defined as $\delta_{sj} = r_{Tsj}-r_{Csj}$. We say that $\boldsymbol{\delta}$ is \textit{compatible} if $\delta_{sj}=0$ whenever $Z_{sj}=1$ and $R_{sj}=0$ or $Z_{sj}=0$ and $R_{sj}=1$. Under this definition, we can express the number of attributable effects as $A=\Z^T\boldsymbol{\delta}$. 
For a compatible $\boldsymbol{\delta}$ such that $\Z^T\boldsymbol{\delta}= a$ we denote $\widetilde{T}_{-\boldsymbol{\delta}}$ to be $\widetilde{T}$ adjusted for the $a$ attributable effects. $\widetilde{T}_{-\boldsymbol{\delta}}$ defines a new reference distribution for $\sum_{s=1}^{S}Z_{s1}r_{Cs1}$ under the null hypothesis that potential accidents indicated by $\boldsymbol{\delta}$ are attributable to talking on a cellphone while driving.  We can write $\widetilde{T}_{-\boldsymbol{\delta}}$ as
\begin{equation}\label{eq:T-delt}
    \widetilde{T}_{-\boldsymbol{\delta}} = \frac{\sum_{s=1}^{S+c^+}Z_{s1}R_{s1}(1-\delta_{s1}) - \left(\sum_{s=1}^S(1-\delta_{s1})p_s+ \sum_{s=S+1}^{S+c^+}(1-\delta_{s1})\right)}{\left\{\sum_{s=1}^S(1-\delta_{s1})p_s(1-p_s)\right\}^{1/2}}\,.
\end{equation}

Using the notion of asymptotic separability (\cite{gastwirth2000}),  \cite{rosenbaum2002a} show that choosing a compatible $\boldsymbol{\delta}^*\equiv\boldsymbol{\delta}^*(a)$ with $\Z^T\boldsymbol{\delta}^*(a)=a$ that maximizes the expectation, and when there are ties to maximize the variance term, yields a reference distribution that, asymptotically, has the largest upper tail area among compatible $\boldsymbol{\delta}(a)$. Thus, we can use $T_{-\boldsymbol{\delta}^*}$ to test the plausibility that there are at most $a$ attributable effects. Since $A$ is a random variable we refrain from calling this a hypothesis test, a term usually reserved for unknown parameters. From equation \eqref{eq:T-delt} we see that $\boldsymbol{\delta}^*(a)$ includes the $a$ pairs in $D(+,-)$ with the smallest values of $p_s$.

It is possible to invert the one-sided ``plausibility tests" introduced above using $T_{-\boldsymbol{\delta}^*}$ that we just introduced in order to construct a confidence interval for attributable effects of the form $\{A:\,A> a\}$. It turns out that if it is plausible that there are $a$ attributable effects then it is also plausible that there are $a+1$ attributable effects (\cite{rosenbaum2002}).
This monotonicity property leads to a very simple procedure to construct a one-sided confidence interval in the absence of hidden bias. First, if $p_s=1/2$ for all $s=1,2,\dots,\widetilde{S}$ then for any $a\ge 0$ we can compute $\widetilde{T}_{-\boldsymbol{\delta}^*}=\{T - a - (S-a)/2\}/\{(S-a)^{1/2}/2\}$.

Next, starting with $a=0$ we check if $\widetilde{T}_{-\boldsymbol{\delta}^*} < \Phi^{-1}(1-\alpha)$, incrementing $a$ by one if it isn't and stopping if it is. Finally, let $a^*$ be equal to one less the value of $a$ at which we terminate the procedure. Using the monotonicity result above we have that $\{A:\,A > a^*\}$ is a one-sided $100\times(1-\alpha)\%$ confidence interval.

If we bound the worst case calibrated bias above by $\Gamma$ then we can construct a one-sided $100\times(1-\alpha)\%$ confidence interval following the same procedure but instead using $\widetilde{T}_{-\boldsymbol{\delta}^*,\Gamma} = \{T - a - (S-a)p_\gamma\}/\{(S-a)p_\gamma(1-p_\gamma)\}^{1/2}$
as our standard deviate where $p_\gamma = \Gamma/(1+\Gamma)$. The resulting one-sided $100\times(1-\alpha)\%$ confidence interval is referred to as a \textit{sensitivity interval} (See Chapter 4, \cite{rosenbaum2002}). For a detailed illustration of these procedures we refer the reader to Sections 3-6 of \cite{rosenbaum2002a}.

\section{FROM WORST CASE TO AVERAGE CASE SENSITIVITY ANALYSIS}\label{worst-avg}
 
\subsection{Valid average case analysis: binary outcome}\label{avg-bin}
An investigator conducting a sensitivity analysis tries to determine a test statistic whose null distribution is known conditional on the presence of hypothetical bias $\Gamma$. Since the distribution of $U_{si}$ is unknown, traditionally, the investigator assumes the worst. That is, the null distribution is constructed assuming that in each pair $u_{s1} = 1$ and $u_{s2}=0$.  As noted in Section \ref{ss-review},  $T^+$ yields a valid reference distribution for testing the null of no-treatment effect when $\Gamma \ge \Gamma_{truth}$. However, such a test is inherently conservative because it is designed to be valid for any realization of $\U$ since $\U$ and thus since $\p = (p_1,\dots,p_S)^T$ are generally unknown. This is why we resort to a sensitivity analysis where we allow $p_s$ to vary arbitrarily as long as $p_s/(1-p_s) \le \Gamma$. 
In Section~\ref{subsec:evidence} we asked whether there was some natural quantification of average bias to which we could calibrate our sensitivity analysis which would lead to a less conservative analysis than the worst case calibration. One such quantification is $\Gamma'_{truth} = \overline{\p}/(1-\overline{\p})$  where $\overline{\p}$ is the sample average of $p_s$. In what follows, we show that if we calibrate our sensitivity analysis to $\Gamma'_{truth}$ it will be valid and less conservative than the worst case calibration.
To prove this, we show that $T' \sim \operatorname{B}(\Gamma'_{truth}/(1+\Gamma'_{truth}),S)$ yields a valid reference distribution for testing the null of no treatment effect against the alternative of a positive treatment effect.
In Theorem \eqref{thm:avgGam} below, we prove that the upper tail probability for McNemar's statistic $T$ is bounded above by the upper tail probability for $T'$.
 \begin{theorem}
   \label{thm:avgGam}
   Set $\overline{\p}= \left(\sum_{s=1}^S p_s\right)/S$ and $\Gamma'_{truth} = \overline{\p}/(1-\overline{\p})$ and let $V_s \stackrel{iid}{\sim} \operatorname{Bern}(\Gamma'_{truth}/(1+\Gamma'_{truth}))$ for all $s=1,2,\dots,S$. Define $T' = V_1 + \dots + V_S$. Then
   \[  \Prob{T \ge a |\mathcal{F},\,\mathcal{Z}\cap \mathcal{A}_{\mathbf{1}}}\le \Prob{T' \ge a}\;\text{for all}\;\; a \ge S\overline{\p}\,.\]
 \end{theorem}
 \begin{proof}
   Observe that $\p$ majorizes $\overline{\p}\cdot\1$ and note that if a function $f(\p)$ is Schur-convex in $\p$ then $f(\p) \ge f(\overline{\p}\1)$. What remains to be shown is that the distribution function for a Poisson-Binomial is Schur-convex in $\p$. See \cite{gleser1975} for this approach and \cite{hoeffding1956} for the original proof. The theorem as stated is an immediate corollary of Theorem 4 in \cite{hoeffding1956}.  \cite{gleser1975} presents a more general version of this result which holds when the success probabilities of $T$ majorize those of $T'$.
 \end{proof}
\begin{remark}\label{rmk:avgGam}
  Theorem \eqref{thm:avgGam} is a finite sample result whose proofs we refer to are both rather technical. An analogous asymptotic result follows from much simpler arguments. The variance of a Bernoulli random variable with success $p$ can be written as $f(p)=p(1-p)$. $f$ is clearly concave and thus by Jensen's Inequality, $\Var{T} \le \Var{T'}$. Since the expectation of $T$ and $T'$ are equal, using a normal approximation to the exact permutation test will asymptotically yield the same stochastic ordering as in Theorem \eqref{thm:avgGam}.
\end{remark}
\begin{remark}\label{rmk:T-plus-bnd}
It is important to note that $\Gamma'_{truth}\le\Gamma_{truth}$ since $p_s/(1-p_s) \le \Gamma_{truth}$ for $s=1,\dots,S$. Consequently, we have that $\Prob{T'\ge a} \le \Prob{T^+\ge a}$ which implies that sensitivity analysis with respect to $\Gamma'$, the average case calibrated sensitivity analysis, is less conservative than the worst case calibrated sensitivity analysis with respect to $\Gamma$.
\end{remark}

The implication of this theorem is that it is safe to interpret a sensitivity analysis in terms of $\Gamma'$, an upper bound on the sample average hidden bias ($\overline{\p}/(1-\overline{\p})$). For example, when using the \textit{most active cellphone day} control window we have $\Gamma_{sens}=2.4$. Previously, we would say that if no case-crossover pair was subject to hidden bias larger than 2.4, then the data would still provide evidence that talking on a cellphone increases the risk of getting in a car accident. Now, some case-crossover pairs may be subject to hidden bias (much) larger than 2.4, as long as the sample average hidden bias is no larger than 2.4. It is important to note that this interpretation is only valid for binary outcomes. The proof relies on Schur-convexity of the distribution function of our test statistic with respect to $\p$ which requires that it be symmetric in $\p$. For more general tests, such as the sign-rank test, this is not the case.

Some additional applications of Theorem~\ref{thm:avgGam} can be found in the Web Appendices. Web Appendix A considers the case when $U_{s1}$ and $U_{s1}$ measure some time-varying propensity of subject $s$ to use his cellphone. Using Theorem~\ref{thm:avgGam} we develop a little theory and a numerical example. Web Appendix B provides details on how Theorem \ref{thm:avgGam} can be applied when $U$ is not restricted to the unit interval.


\section{THE EFFECT OF CELLPHONE USE ON MOTOR-VEHICLE COLLISIONS}\label{examples}

In this section we return to our motivating example to see how our average case theory can provide interpretive assistance to our standard sensitivity analysis we carried out in Section \ref{sens-examp} and allow us to incorporate additional information to empirically calibrate our average case sensitivity analysis.

\subsection{Driving intermittency}\label{subsec:inter}

The study conducted in \cite{tibshirani1997} did not have access to direct information on whether an individual was driving during the control window. The authors examine the effect of driving intermittency during the control window on their relative-risk estimate by bootstrapping the estimate using an intermittency rate of $\widehat{\rho} = 0.65$. In other words, they correct for bias due to the possibility that a subject was not driving during the control window. The intermittency rate was estimated using a survey asking 100 people who reported car crashes whether they were driving at the same time the previous day. Alternatively, one may ask a related question in the context of a sensitivity analysis - does the bias due to driving intermittency explain the observed association between cellphone usage and traffic incidents? Given that the study took place in the early 1990s when, for some cellphones and carphones were synonymous, it would not be surprising if many study participants (almost) always used their landlines rather than their cellphones when not driving, violating the positivity assumption. Therefore, the only plausible $\Gamma_{truth}$ is infinite (or at least very large) when conditioning on case-crossover pairs where the subject is on her cellphone in only one of the two windows. This renders the worst case sensitivity analysis uninformative. No magnitude of association between cellphone use and car accidents would convince us that the relationship was causal if we stuck to the worst case calibration of the sensitivity analysis. The average case calibration, on the other hand, still has a chance. We can use our estimate $\widehat{\rho}$ to approximate a plausible value of $\overline{\textbf{p}}$,
$\overline{\textbf{p}} = (1-\widehat{\rho})\cdot 1 + \widehat{\rho}\cdot 0.5 = 0.675$, 
and a corresponding plausible value of $\Gamma'_{truth}$, 
$\Gamma'_{truth} = \overline{\textbf{p}}/(1-\overline{\textbf{p}}) = 2.1 \,.$
Theorem \eqref{thm:avgGam} circumvents the conceptual hurdle of unbounded $\Gamma_{truth}$ and allows us to confidently use a sensitivity analysis to quantitatively assess the causal evidence. Moreover, it allows us to incorporate information about $\rho$ into our analysis. If the association between cellphone use and motor vehicle collisions is causal in nature, our empirical calibration suggests that our test for treatment effect should be insensitive to unobserved biases with magnitude $\Gamma' \approx 2.1$.

\subsection{An alternative approach to handling pairs with unbounded bias}
There are other approaches to dealing with the example of infinite bias we just presented. For instance, the investigator may be more confident in specifying an upper bound on the worst case bias to be finite, $\Gamma<\infty$, for a proportion $1-\beta$ of the matched pair sample than he is in working in terms of the average case bias. If he has a good sense of what proportion $\beta$ of the pairs is exposed to unbounded bias he may drop $\beta\times S$ pairs where the treated unit had positive outcome and perform the standard worst case sensitivity analysis on the remaining $(1-\beta)\times S$ pairs. \cite{rosenbaum1987} proved that this method yields a valid sensitivity analysis. This strategy would be particularly suited for the example of driver intermittency discussed above. However, this approach assumes this particular pattern of unmeasured confounding is present and driver intermittency is just one of many sources of potential bias. On the other hand, the average case analysis accomodates arbitrary patterns of bias that may lead to large differences in average and worst case biases.

\subsection{Average case sensitivity analysis for attributable effects}\label{attr-sec}

How many of the recorded accidents in our study can be attributed to the driver talking on a cellphone? Recall from Section \ref{subsec:attr-binary} that the set indicated by $\boldsymbol{\delta}^*$ includes the $a$ pairs in $D(+,-)$ with the smallest values of $p_s$. Although we cannot compute $\widetilde{T}_{-\boldsymbol{\delta}^*}$ and thus cannot use it directly to conduct inference, we can compute a lower bound that we will show can be used to perform an average case sensitivity analysis:

\begin{align}\label{eq:norm-dev-star}
  \widetilde{T}_{-\boldsymbol{\delta}^*} & = \frac{\sum_{s=1}^{S}Z_{s1}r_{Cs1} - \sum_{s=1}^S(1-\delta^*_{s1})p_s}{\left\{\sum_{s=1}^S(1-\delta^*_{s1})p_s(1-p_s)\right\}^{1/2}} \notag \\
  & = \frac{\sum_{s=1}^S Z_{s1}R_{s1}(1-\delta^*_{s1})  - \sum_{s=1}^S(1-\delta_{s1}^*)p_s }{\left\{\sum_{s=1}^S(1-\delta^*_{s1})p_s(1-p_s)\right\}^{1/2}} \notag \\
  & = \frac{T-a - (S-a)\overline{\p}(a) }{\left\{\sum_{s=1}^S(1-\delta^*_{s1})p_s(1-p_s)\right\}^{1/2}} \notag\\
  & \ge \frac{T-a - (S-a)\overline{\p}(a) }{\left\{(S-a)\overline{\p}(a)(1-\overline{\p}(a))\right\}^{1/2}} = \widetilde{T}(\overline{\p}(a))
\end{align}
where $\overline{\p}(a) = \sum_{s=1}^S(1-\delta_{s1}^*)p_s/(S-a)$. The last inequality follows from Jensen's inequality applied to the variance term in the denominator.  Notice that instead of applying Theorem \eqref{thm:avgGam} in order to derive a sensitivity analysis in terms of the average bias we use the simpler argument in Remark \eqref{rmk:avgGam}. Now note that if $p_s \ge \underline{p}$ for all $s=1,\dots,S$ then we can relate the trimmed average probability, $\overline{\p}(a)$, to $\overline{\p}$ as follows
\begin{equation}\label{p-bar-lb} \overline{\p} \ge \frac{(S-a)\overline{\p}(a)+a\cdot\underline{p}}{S} = q(a)\,.\end{equation}

We can use this relationship to construct a simple procedure -- mirroring that of Section \ref{subsec:attr-binary} -- to perform an average case calibrated sensitivity analysis for one-sided confidence intervals of the form $\{A:\,A> a\}$ that yields average case calibrated sensitivity intervals. The procedure can be summarized as follows,

\begin{enumerate}[leftmargin=20pt]
  \item{} Choose a desired average calibrated sensitivity parameter $\Gamma'$.
  \item{} For $a=0$ solve $q(a) = \Gamma'/(1+\Gamma')$ for $\overline{\p}(a)$ and denote the solution $p(a,\gamma')$. Compute $\widetilde{T}(p(a,\gamma'))$.
  \item{} If $\widetilde{T}(p(a,\gamma')) < \Phi^{-1}(1-\alpha)$ then conclude it is plausible that none of the accidents can be attributed to talking on a cellphone.
  \item{} Else, repeat steps (2) and (3) for $a=1,\dots,S$ stopping when $\widetilde{T}(p(a,\gamma')) < \Phi^{-1}(1-\alpha)$. Let $a^* = a-1$.
  \item{} Return the $100\times(1-\alpha)\%$ sensitivity interval $\{A:\,A > a^*\}$ and conclude that it is plausible that more than $a^*$ of the accidents are attributable to talking on a cellphone when exposed to an average bias of at most $\Gamma'$.
\end{enumerate}

Just as in the simple test for no treatment effect, we see that we have a nearly identical procedure to the worst case sensitivity analysis with an average interpretation of the bias parameter. In fact, the procedure also yields a corresponding worst case calibration for the computed sensitivity interval. Under the worst case calibration, the sensitivity interval from step (5) would correspond to a worst case bias $\Gamma = p(a^*,\gamma')/(1-p(a^*,\gamma'))$.

 How might we apply this procedure to our example? For a given control window we would like to make confidence statements such as, \textit{at the 95\% level it is plausible that there are $a^*$ or more accidents attributable to talking on a cellphone}. Recall the empirically calibrated average case bias from Section \ref{subsec:inter}, $\Gamma' \approx 2.1$. We may also be interested making sensitivity statements such as, \textit{if the average probability of talking on a cellphone during the hazard window is at most 2.1 times that of talking on a cellphone in the control window for drivers in our study, $\Gamma'=2.1$, it is plausible at the 95\% level that there are $a^*$ or more accidents attributable to talking on a cellphone}. Table \ref{attr-table} summarizes the plausible range of attributable accidents for each of the four different control windows. For all four control windows we set $\Gamma'=2.1$. The first column is the number of discordant pairs in which the driver was on a cellphone during the control window. The second column reports the lower bound $a^*$ of the one-sided sensitivity intervals for $\alpha=0.05$ We also report the corresponding worst case calibrated bias in the last column of Table \ref{attr-table}. In the cellphone study we have no convincing reason to believe that $\underline{p} > 0$ but in other examples, it may make sense that $p_s$ is bounded from below, which has the effect of making the procedure less conservative.

\begin{table}[ht]
\begin{center}
\begin{tabular}{lcccr} 
  Control Window & $|D(+,-)|$ & $a^*$ & $\Gamma'$ & $\Gamma$\\ 
  \hline
  previous weekday/weekend &158&28&2.1&4.04\\ 
  one week prior &164&31& 2.1&4.37 \\
  previous driving day &119&18&2.1&3.51 \\  
  most active cellphone day &134&5&2.1&2.3
\end{tabular}
\vspace{20pt}
\caption{Sensitivity analysis for 95\% one-sided confidence intervals for attributable effects of the form $\{A:\,A> a^*\}$. $\Gamma'$ indicates the average calibration bias that we specify for the procedure and $\Gamma$ is the implied worst case calibration that corresponds to the computed interval. We assume that $\underline{p}=0$.}\label{attr-table}
\end{center}
\end{table}

We find that even if the average probability of talking on a cellphone during the hazard window was at most 2.1 times that of talking on a cellphone on the same day one week prior, it is plausible that there are 31 or more accidents attributable to talking on a cellphone. The implied worst case bias associated with this statement is $\Gamma=4.37$. What this means is that we would arrive at the same conclusion about the number of plausible attributable accidents if we put an upper bound on the worst case bias of $\Gamma=4.37$ and followed the standard confidence interval procedure for attributable effects outlined in Section \ref{subsec:attr-binary} and \cite{gastwirth2000}. Unlike the sensitivity analysis for the simple test for no treatment effect, the average case calibrated sensitivity analysis for attributable effects is not guaranteed to be less conservative than the worst case calibration. For a 95\% sensitivity interval for attributable effects generated by our procedure where $a^* >0$, the corresponding upper bound on the average case bias $\Gamma'$ is less than the corresponding upper bound on the worst case bias $\Gamma$. This occurs since we do not know which pairs contain attributable effects nor do we know each pair's particular exposure to hidden bias. Without any further assumptions, the best lower bound for $\Gamma'$ assumes that all the $a$ pairs with attributable effects have arbitrarily small probability of being on a cellphone in the hazard window and not the control window. This is expressed mathematically in equation \eqref{p-bar-lb} by setting $\underline{p}=0$. If $\Gamma'_{truth}<\Gamma_{truth}$ -- which is a reasonable assumption in most circumstances -- then using the average case calibration may still result in a less conservative analysis. However, if all case-crossover pairs are exposed to the same magnitude of bias such that $\Gamma'_{truth}=\Gamma_{truth}$ then we are guaranteed to be less conservative by using the worst case calibration. A reasonable solution would be to simply supply both the $\Gamma'$ and $\Gamma$ when reporting a sensitivity interval, as we do in Table \ref{attr-table}. The investigator may then present an argument based on subject matter expertise as to which calibration is likely to be less conservative.

\section{DISCUSSION} 
The theorem presented in \ref{avg-bin} can be thought of as an interpretive aid: For the same standard sensitivity analysis we now have an additional, often more natural, way to interpret the results. This new average case interpretation may also allow researchers to make use of additional information about the problem to empirically calibrate their sensitivity analysis. As we saw in Section \ref{subsec:inter}, we used the estimate of driver intermittency rate to determine an approximate lower bound on $\Gamma'_{truth}$, providing us with some empirical guidance when conducting our sensitivity analysis. In the worst case setting, such an empirical calibration would not be possible. The investigator performs a sensitivity analysis in anticipation of critics who might claim the association is due to some unobserved confounder. The average case analysis makes the protection that the sensitivity analysis provides against such criticism more robust. As the title of the article makes clear, the results we present are for binary data. As we illustrated in Section \ref{attr-sec}, the notion of attributable effects allows us to construct interpretable confidence intervals for binary outcomes. We show that our average case calibration can be extended to the sensitivity analysis of such confidence intervals and in most cases will yield a less conservative conclusions. It may then be interesting to apply the results here to the sensitivity analysis of displacement effects, the continuous analog of attributable effects for non-binary outcomes. \cite{rosenbaum2002a} show that displacement effects can be analyzed in the attributable effect framework for binary response, providing a potential avenue to extend average case calibrated sensitivity analysis to a study with non-binary outcomes. 

\section{Supplementary Materials}
Web Appendices A and B referenced in Section \ref{worst-avg} and the R code that produced the sensitivity analysis summarized in Table \ref{sens-table} in Section \ref{sens-examp}, the Monte Carlo simulation found in Web Table 1 in Web Appendix A, and the attributable effects analysis in Section \ref{attr-sec} are available with this paper at the Biometrics website on Wiley Online Library.

\section*{Acknowledgements} Raiden Hasegawa would like to thank Colin Fogarty for his insightful comments and feedback.\vspace*{-8pt}

\bibliographystyle{biom}
\bibliography{./biblio} 

\label{lastpage}

\end{document}